\documentclass[letterpaper, 11pt]{article}
\usepackage{hyperref}
\usepackage{amsmath,amssymb,amsthm}
\usepackage{bm}

\newtheorem{theorem}{Theorem}
\newtheorem{lemma}{Lemma}
\newtheorem{corollary}[lemma]{Corollary}
\newtheorem{proposition}[lemma]{Proposition}
\theoremstyle{definition}
\newtheorem{definition}[lemma]{Definition}

\newtheorem{fact}[lemma]{Fact}
\theoremstyle{remark}
\newtheorem{remark}[lemma]{Remark}

\newcommand{\sign}{\mathrm{sign}}

\title{Better Protocol for XOR Game using Communication Protocol and Nonlocal Boxes}
\author{Ryuhei Mori\thanks{Supported by MEXT KAKENHI Grant Number JP24106008 and JSPS KAKENHI Grant Number JP17K17711.}}

\date{%
School of Computing, Tokyo Institute of Technology, Tokyo, Japan.\\
\texttt{mori@c.titech.ac.jp}\\
}

\begin{document}
\maketitle
\begin{abstract}
\normalsize
Buhrman showed that an efficient communication protocol implies a reliable XOR game protocol.
This idea rederives Linial and Shraibman's lower bound of randomized and quantum communication complexities, which was derived by using factorization norms, with worse constant factor in much more intuitive way.
In this work, we improve and generalize Buhrman's idea, and obtain a class of lower bounds for randomized communication complexity
including an exact Linial and Shraibman's lower bound as a special case.
In the proof, we explicitly construct a protocol for XOR game from a randomized communication protocol
by using a concept of nonlocal boxes and Paw\l owski et~al.'s elegant protocol, which was used for showing the violation of information causality in superquantum theories.
\end{abstract}

\vspace{1em}
\noindent
\textbf{Keywords:} Communication complexity, XOR game, nonlocality, CHSH inequality, Fourier analysis.

\thispagestyle{empty}
\clearpage
\setcounter{page}{1}

\section{Introduction}
Communication complexity is one of the central tool in theoretical computer science.
In this work, we investigate an extremely simple technique for lower bounding communication complexity and XOR-amortized communication complexity.
Let $C_{\rho}(f,\mu)$ be a deterministic communication complexity for computing $f\colon\{0,1\}^n\times\{0,1\}^n\to\{0,1\}$ with an error probability at most $(1-\rho)/2$ on an input distribution $\mu$.
Let $\beta(f,\mu)$ be the largest bias of an XOR game for $f\colon\{0,1\}^n\times\{0,1\}^n\to\{0,1\}$ on an input distribution $\mu$, i.e., the largest winning probability of the XOR game is $(1+\beta(f,\mu))/2$.
Buhrman gave a clear argument on a relationship of $C_\rho(f,\mu)$ and $\beta(f,\mu)$~\cite{TCS-040, briet2013multipartite}.
\begin{fact}[Buhrman's argument a.k.a.\ discrepancy bound]\label{fa:briet}
For any $\rho\in[0,1]$,
\begin{equation*}
\beta(f, \mu) \ge \rho 2^{-C_{\rho}(f, \mu)}.
\end{equation*}
\end{fact}
\begin{proof}
We will construct a protocol for the XOR game of $f$ with the bias $\rho 2^{-C_\rho(f,\mu)}$ using the communication protocol $P$ corresponding to $C_\rho(f,\mu)$.
Alice and Bob use $C_\rho(f,\mu)$ shared random bits, and regard them as a transcript of the communication protocol $P$.
Alice and Bob check the consistency at each side, and output a uniform random bit if they are inconsistent.
If the random bits are consistent at each side, Alice and Bob output the output of the communication protocol $P$ and 0, respectively.

If the shared random bits are inconsistent transcript at least one side, the XOR of their output is a uniform random bit.
If the shared random bits are consistent transcript at the both sides, the XOR of their output is the output of the communication protocol $P$, which is $f(x,y)$ with probability at least $(1+\rho)/2$.
Hence, the winning probability of this protocol for the XOR game is at least $2^{-C_\rho(f,\mu)}(1+\rho)/2 + (1-2^{-C_\rho(f,\mu)})(1/2) = (1+\rho 2^{-C_\rho(f,\mu)})/2$.
\end{proof}
Buhrman's argument gives a lower bound $C_{\rho}(f, \mu)\ge \log\frac{\rho}{\beta(f,\mu)}$ of the communication complexity, which is equivalent to the well-known discrepancy bound up to an additive constant.
Buhrman's argument gave an operational meaning to the discrepancy bound.
Furthermore, the above argument with small modification gives Linial and Shraibman's lower bound, which was obtained by using factorization norms,
 with a worse constant factor~\cite{linial2009lower, TCS-040}.
This argument can be straightforwardly generalized to a quantum setting~\cite{TCS-040, briet2013multipartite}.
Let $C^*_{\rho}(f,\mu)$ be a quantum communication complexity, which is the number of bits transmitted (not qubits), for computing $f\colon\{0,1\}^n\times\{0,1\}^n\to\{0,1\}$
using shared quantum states with an error probability at most $(1-\rho)/2$ on an input distribution $\mu$.
Let $\beta^*(f,\mu)$ be the largest bias of an XOR game using shared quantum states for $f\colon\{0,1\}^n\times\{0,1\}^n\to\{0,1\}$ on an input distribution $\mu$.
Then, a straightforward generalization of Fact~\ref{fa:briet} shows $C_{\rho}^*(f, \mu)\ge \log\frac{\rho}{\beta^*(f,\mu)}$.
In the following, we demonstrate how to show $\max_\mu C^*_\rho(\mathrm{IP}_n,\mu)= \Theta(n)$ for any constant $\rho\in(0,1]$ where $\mathrm{IP}_n(x,y):=\langle x,y\rangle:=\bigoplus_{i=1}^n x_i\wedge y_i$.
Let $f^{\oplus \ell}(x_1,\dotsc,x_{\ell n}):=f(x_1,\dotsc,x_n)\oplus\dotsb\oplus f(x_{(\ell-1)n+1},\dotsc,x_{\ell n})$,
and $\mu^{\otimes \ell}(x_1,\dotsc,x_{\ell n}):=\mu(x_1,\dotsc,x_n)\times\dotsm\times\mu(x_{(\ell-1)n+1},\dotsc,x_{\ell n})$.
Cleve et al.\ showed that the XOR game in the quantum physics satisfies the perfect parallel repetition theorem, i.e., $\beta^*(f^{\oplus \ell}, \mu^{\otimes \ell}) = \beta^*(f,\mu)^\ell$~\cite{Cleve2008}.
From the perfect parallel repetition theorem, we obtain
\begin{equation*}
\min_\mu \beta^*(\mathrm{IP}_n,\mu)\le
\min_\nu \beta^*(\mathrm{IP}_n,\nu^{\otimes n})= \min_\nu \beta^*(\mathrm{AND},\nu)^n 
=2^{-\frac{n}2}.
\end{equation*}
The last equality is obtained by the Tsirelson bound~\cite{cirelson1980} (In fact, the above inequality is equality. It is easy to see that the worst input distribution for $\mathrm{IP}_n$ is the uniform distribution).
Hence, we obtain
\begin{equation*}
\max_\mu C^*_{\rho}(\mathrm{IP}_n, \mu) \ge \log\frac{\rho}{\min_\mu \beta^*(\mathrm{IP}_n,\mu)}\ge \frac{n}2 + \log \rho.
\end{equation*}
This bound was obtained by Kremer~\cite{kremer1995quantum} and Linial and Shraibman~\cite{linial2009lower}.
However, the above derivation is extremely simple and intuitive, and only needs Buhrman's argument, Cleve et al.'s perfect parallel repetition theorem and the Tsirelson bound.

In this work, we improve Buhrman's argument, Fact~\ref{fa:briet}, for deterministic and randomized communication complexities by using quantum theory.
First, we obtain
\begin{equation*}
\beta^*(f, \mu) \ge \rho 2^{-\frac12C_{\rho}(f, \mu)}
\end{equation*}
or equivalently
\begin{equation*}
C_{\rho}(f, \mu) \ge 2\log\frac{\rho}{\beta^*(f, \mu)}.
\end{equation*}
Note that this lower bound is worse than Linial and Shraibman's lower bound~\cite{linial2009lower}, but is easier to evaluate.
Since $\beta^*(f,\mu)\ge\beta(f,\mu)$, an improvement from the previous lower bound $\log\frac{\rho}{\beta(f,\mu)}$ is at most a factor 2.
Although this improvement is typically not significant, in this work, we obtain non-trivial lower bounds of XOR-amortized communication complexities
of the equality function.

\begin{theorem}\label{thm:exm}
For any $\rho\in(0,1]$,
\begin{align*}
\lim_{n\to\infty} \lim_{\ell\to\infty} \max_{\mu} \frac1\ell C_{\rho^\ell}(\mathrm{EQ}_n^{\oplus\ell}, \mu^{\otimes\ell})&\ge 2\log 3 + 2\log \rho\\
\lim_{n\to\infty} \lim_{\ell\to\infty} \max_{\nu} \frac1\ell C_{\rho^\ell}(\mathrm{EQ}_n^{\oplus\ell}, (\nu^{\otimes n})^{\otimes\ell})&\ge 2\log \left(1+\frac{2}{\sqrt{3}}\right) + 2\log \rho
\end{align*}
where $\mu$ and $\nu$ are probability distributions on $\{0,1\}^n\times\{0,1\}^n$ and $\{0,1\}\times\{0,1\}$, respectively.
Here, $2\log 3 \approx 3.1699$ and $2\log(1+2/\sqrt{3})\approx 2.2150$.
\end{theorem}
Since randomized/distributional communication complexity of the equality function is constant~\cite{kushilevitz1997communication}, these improvements are meaningful.
If the factor 2 is missing, the above two lower bounds for the equality function are smaller than 2 for $\rho=1$.
In this sense, this improvement is significant.

By applying an argument for generalized discrepancy theory~\cite{briet2013multipartite, TCS-040}, we can further improve this bound, and obtain 
a class of lower bounds for randomized communication complexity including an exact Linial and Shraibman's lower bound as a special case.
\begin{theorem}[Linial and Shraibman's lower bound~\cite{linial2009lower, TCS-040}]\label{thm:ls}
For any $f\colon\{0,1\}^n\times\{0,1\}^n\to\{0,1\}$ and $\epsilon\in[0,1/2]$,
\begin{equation}
R_\epsilon(f)\ge 2  \log\max_{h, \mu}\frac{(1-\epsilon)\mathbb{E}_{(x,y)\sim \mu}[f(x,y)h(x,y)]-\epsilon}{\beta^*(h,\mu)}
\label{eq:ls}
\end{equation}
where $R_\epsilon(f)$ is the randomized communication complexity of $f$ with an error probability at most $\epsilon$, where $h\colon\{0,1\}^n\times\{0,1\}^n\to\{0,1\}$, and where
$\mu$ is a probability distribution on $\{0,1\}^n\times\{0,1\}^n$.
\end{theorem}
While the original proof of Theorem~\ref{thm:ls} uses factorization norms of matrices,
the proof in this paper is based on Buhrman's argument, which derived Theorem~\ref{thm:ls} in an intuitive way without the constant factor 2~\cite{TCS-040, briet2013multipartite}.
For obtaining the constant factor 2, we use a concept of nonlocal boxes and Paw\l owski et al.'s idea, which showed the violation of information causality in superquantum theories~\cite{pawlowski2009information}.
In the generalized lower bounds, the constant factor 2 is replaced by larger constants while $\beta^*(h,\mu)$ is also replaced by larger quantities.
Although any concrete lower bound by the generalized lower bound is not obtained in this paper, the generalized lower bound may improve Linial and Shraibman's lower bound at most a constant factor.

\section{Preliminaries}
\subsection{Nonlocal box}
The nonlocal box is an abstract device with two input ports and two output ports.
When $x\in\{0,1\}$ and $y\in\{0,1\}$ are given to a nonlocal box, 
a nonlocal box randomly outputs $a\in\{0,1\}$ and $b\in\{0,1\}$.
A nonlocal box is specified by a conditional probability distribution $p(a,b\mid x, y)$.
Here, a nonlocal box is an abstract device representing an ``entangled state'' where $x$ and $y$ correspond to a choice of ``measurements'', and where
$a$ and $b$ correspond to ``outcomes'' of the measurements $x$ and $y$, respectively.
Hence, the conditional probability distribution must satisfy the no-signaling condition
\begin{align*}
\sum_{b\in\{0,1\}} p(a,b\mid x,0) &= \sum_{b\in\{0,1\}} p(a,b\mid x,1)\\
\sum_{a\in\{0,1\}} p(a,b\mid 0,y) &= \sum_{a\in\{0,1\}} p(a,b\mid 1,y)
\end{align*}
since if the no-signaling condition is violated, two distant parties can communicate only by measuring a shared state, which is a communication faster than light, and must be forbidden.
Let the CHSH probability be
\begin{equation*}
P_{\mathrm{CHSH}} := \frac14\sum_{\substack{a\in\{0,1\},\,b\in\{0,1\},\\x\in\{0,1\},\,y\in\{0,1\}\\ a\oplus b = x\wedge y}} P(a,b\mid x,y).
\end{equation*}
When
\begin{equation*}
p(a,b\mid x,y) =
\begin{cases}
P_{\mathrm{CHSH}}/2,& \text{if } a\oplus b = x\wedge y\\
(1-P_{\mathrm{CHSH}})/2,& \text{otherwise}
\end{cases}
\end{equation*}
the nonlocal box is said to be isotropic.
The CHSH bias $2P_{\mathrm{CHSH}}-1$ is denoted by $\delta$, i.e., $P_{\mathrm{CHSH}} = (1+\delta)/2$.
Classical physics and quantum physics can simulate isotropic nonlocal boxes with CHSH bias up to $1/2$~\cite{PhysRevLett.23.880} and $1/\sqrt{2}$~\cite{cirelson1980}, respectively. 

\subsection{Communication complexities}
Let $C_{\rho}(f,\mu)$ be a deterministic communication complexity for computing $f\colon\{0,1\}^n\times\{0,1\}^n\to\{0,1\}$ with an error probability at most $(1-\rho)/2$ on an input distribution $\mu$.
Let $C^*_{\rho}(f,\mu)$ be a quantum communication complexity,  which is the number of bits transmitted (not qubits), for computing $f\colon\{0,1\}^n\times\{0,1\}^n\to\{0,1\}$ using shared quantum states with an error probability at most $(1-\rho)/2$ on an input distribution $\mu$.
Let $C_{\mathrm{NL}(\delta),\rho}(f,\mu)$ be a communication complexity with isotropic nonlocal boxes with CHSH bias $\delta\ge 1/2$.
Let $R_\epsilon(f)$ be a randomized communication complexity with an error probability at most $\epsilon\in[0,1/2]$.

\begin{remark}
In this paper, all $C_\rho(f, \mu)$ (and its variants) can be replaced by
$C^\oplus_\rho(f, \mu)$ which is the communication complexity for computing $a$ and $b$ by Alice and Bob, respectively such that $a\oplus b = f(x,y)$.
Since the difference between $C_\rho(f,\mu)$ and $C_\rho^\oplus(f,\mu)$ is at most two, and is negligible for the amortized case, in this paper, we use $C_\rho(f,\mu)$ for the simplicity.
\end{remark}

\subsection{XOR game}
In a two-player XOR game $(f,\mu)$, Alice and Bob are given $x\in\{0,1\}^n$ and $y\in\{0,1\}^n$ according to the input distribution $\mu$, and output $a\in\{0,1\}$ and $b\in\{0,1\}$, respectively
without communication for computing $f\colon\{0,1\}^n\times\{0,1\}^n\to\{0,1\}$.
Alice and Bob win if and only if $a\oplus b=f(x,y)$.
Let $\beta(f,\mu)$ be the largest bias (of the winning probability) of an XOR game for $f$ on an input distribution $\mu$, i.e., the largest winning probability of the XOR game
$(f,\mu)$ is $(1+\beta(f,\mu))/2$.
Let $\beta^*(f,\mu)$ be the largest bias of an XOR game using shared quantum states.
Let $\beta_{\mathrm{NL}(\delta)}(f,\mu)$ be the largest bias of an XOR game using isotropic nonlocal boxes with CHSH bias $\delta\ge 1/2$.
If $\mu$ is omitted, we assume the worst input distribution, e.g., $\beta(f):=\min_\mu\beta(f,\mu)$.
It is straightforward to generalize Fact~\ref{fa:briet} to quantum and nonlocal box settings~\cite{briet2013multipartite}.

\begin{fact}\label{fa:briet1}
For any $\rho\in[0,1]$ and $\delta\in[1/2,1]$,
\begin{align*}
\beta_{\mathrm{NL}(\delta)}(f,\mu) &\ge \rho 2^{-C_{\mathrm{NL}(\delta),\rho}(f,\mu)}\\
\beta^*(f,\mu) &\ge \rho 2^{-C^{*}_{\rho}(f,\mu)}.
\end{align*}
\end{fact}

\subsection{Fourier analysis}
Let $A\colon \{+1,-1\}^n\to\mathbb{R}$.
Let $\mathbb{E}[A(z)]:=\frac1{2^n}\sum_{z\in\{+1,-1\}^n}A(z)$.
Let $\widehat{A}(S):=\mathbb{E}[A(z)\prod_{i\in S} z_i]$ for any $S\subseteq[n]:=\{1,\dotsc,n\}$.
Here, $\widehat{A}(S)$ is called a Fourier coefficient.
When we consider Fourier coefficients of boolean function $\{0,1\}^n\to\{0,1\}$, $0$ and $1$ are replaced by $+1$ and $-1$, respectively~\cite{odonnell2014analysis}.
Let $\|\widehat{A}\|_1:=\sum_{S\subseteq[n]} |\widehat{A}(S)|$,
$\|\widehat{A}\|_\infty:=\max_{S\subseteq[n]} |\widehat{A}(S)|$ and
$\|\widehat{A}\|_0:=|\{S\subseteq[n]\mid \widehat{A}(S)\ne 0\}|$.

\section{Main theorems}
In~\cite{PhysRevA.94.052130},
it was shown that
\begin{equation*}
\beta_{\mathrm{NL}(\delta)}(f)\ge \delta^{C_\rightarrow(f)}
\end{equation*}
where $C_\rightarrow(f)$ is a zero-error one-way communication complexity of $f$.
In this paper, we improve the above inequality by using a two-way communication complexity.
\begin{theorem}\label{thm:main0}
For any $\rho\in[0,1]$, $\delta\in[1/2,1]$,
\begin{align*}
\beta_{\mathrm{NL}(\delta)}(f,\mu) &\ge \rho\delta^{C_{\rho}(f,\mu)}.
\end{align*}
Hence, for any $\rho\in(0,1]$ and $\delta\in[1/2,1)$,
\begin{align*}
C_{\rho}(f,\mu) &\ge  \frac1{\log\delta^{-1}} \log\frac{\rho}{\beta_{\mathrm{NL}(\delta)}(f,\mu)}.
\end{align*}
\end{theorem}
Theorem~\ref{thm:main0} seems to be similar to Fact~\ref{fa:briet1}.
Theorem~\ref{thm:main0} and Fact~\ref{fa:briet1} are generalizations of different types of Fact~\ref{fa:briet}.
In Fact~\ref{fa:briet1}, the communication complexities are replaced by those in stronger theories while the bases in the second factor remain $1/2$.
On the other hand, 
in Theorem~\ref{thm:main0}, the base $1/2$ of the exponent $C_\rho(f, \mu)$ is improved to $\delta$ while the communication complexity remains deterministic.
Especially for $\delta=1/\sqrt{2}$, we obtain
\begin{align*}
C_{\rho}(f,\mu) &\ge  2\log\frac{\rho}{\beta^*(f,\mu)}
\end{align*}
from $\beta^*(f,\mu)\ge \beta_{\mathrm{NL}(1/\sqrt{2})}(f,\mu)$.
From the perfect parallel repetition theorem in quantum physics~\cite{Cleve2008, lee2008direct}, the following corollary is obtained.
\begin{corollary}\label{cor:main}
For any $\ell\in\mathbb{N}$ and $\rho\in(0,1]$,
\begin{align*}
\frac1\ell C_{\rho^\ell}(f^{\oplus\ell},\mu^{\otimes\ell}) &\ge  2\log\frac{\rho}{\beta^*(f,\mu)}.
\end{align*}
\end{corollary}

The following theorem with Corollary~\ref{cor:main} gives Theorem~\ref{thm:exm}.
\begin{theorem}\label{thm:main1}
\begin{align*}
\lim_{n\to\infty} \min_{\mu}\beta^*(\mathrm{EQ}_n,\mu) &= \frac13\\
\lim_{n\to\infty} \min_{\nu}\beta^*(\mathrm{EQ}_n,\nu^{\otimes n}) &= 2\sqrt{3}-3
\end{align*}
where $\mu$ and $\nu$ are probability distributions on $\{0,1\}^n\times\{0,1\}^n$ and $\{0,1\}\times\{0,1\}$, respectively.
\end{theorem}

Furthermore, by applying generalized discrepancy theory~\cite{briet2013multipartite, TCS-040},
we can obtain lower bounds for randomized communication complexity.
\begin{theorem}\label{thm:main2}
For any $h\colon\{0,1\}^n\times\{0,1\}^n\to\{0,1\}$, $\epsilon\in[0,1/2]$, and $\delta\in[0,1/2]$,
\begin{equation}\label{eq:main2}
\beta_{\mathrm{NL}(\delta)}(h,\mu) \ge \delta^{R_{\epsilon}(f)}\left((1-\epsilon)\mathbb{E}_{(x,y)\sim \mu}[f(x,y)h(x,y)]-\epsilon\right).
\end{equation}
Hence, if $(1-\epsilon)\mathbb{E}_{(x,y)\sim \mu}[f(x,y)h(x,y)]-\epsilon > 0$,
\begin{align*}
R_\epsilon(f)\ge \frac1{\log\delta^{-1}}  \log\max_{h, \mu}\frac{(1-\epsilon)\mathbb{E}_{(x,y)\sim \mu}[f(x,y)h(x,y)]-\epsilon}{\beta_{\mathrm{NL}(\delta)}(h,\mu)}.
\end{align*}
\end{theorem}
For $\delta=1/\sqrt{2}$, we obtain~Theorem~\ref{thm:ls},
which is exactly same as Linial and Shraibman's lower bound~\cite{linial2009lower, TCS-040},
\begin{align*}
R_\epsilon(f)\ge 2\log\gamma_2^{1/(1-2\epsilon)}-2\log\frac1{1-2\epsilon}
\end{align*}
where $\gamma_2^\alpha$ is some approximate norm of a communication matrix of $f$~\cite{linial2009lower, TCS-040}.
Theorem~\ref{thm:main2} gives an intuitive proof of Linial and Shraibman's lower bound
and generalizations of Linial and Shraibman's lower bound by using nonlocal boxes.
It is not necessarily easy to upper bound $\beta_{\mathrm{NL}(\delta)}(f,\mu)$.
However, similarly to the relationship $\beta^*(f,\mu)\ge \beta_{\mathrm{NL}(1/\sqrt{2})}(f,\mu)$,
some relaxation may give an upper bound of $\beta_{\mathrm{NL}(\delta)}(f,\mu)$.
Note that Brassard et al.\ showed that $\beta_{\mathrm{NL}(\delta)}(f)$ is lower bounded by a positive constant for any $f$ if $\delta>\sqrt{2/3}$~\cite{PhysRevLett.96.250401}.
In the area of foundation of quantum physics,
it is a well-known open problem of whether $\beta_{\mathrm{NL}(\delta)}(f)$ is lower bounded by a positive constant for any $f$ for $\delta\in(1/\sqrt{2},\sqrt{2/3}]$~\cite{RevModPhys.82.665, PhysRevA.94.052130}.

\section{Proofs of Theorem~\ref{thm:main0} and \ref{thm:main2}: Paw\l owski et al.'s protocol}
\subsection{Intuition on the proof of Theorem~\ref{thm:main0}}
The proof of Theorem~\ref{thm:main0} is the most important part in this paper.
In the proof of Fact~\ref{fa:briet}, a transcript of communication protocol is ``guessed'' by uniform random bits, which succeeds with probability at least $2^{-k}$ where $k$ denotes the length of the longest transcript.
In the proof of Theorem~\ref{thm:main0}, a correct transcript of communication protocol is ``selected'' by using isotropic nonlocal boxes with CHSH bias $\delta$.
This ``selection'' can be implemented by a chain of 1-bit selectors $\mathrm{Addr}_1(x_0,x_1,y):=x_y$, which has an XOR game protocol with bias $\delta$.
We obtain Theorem~\ref{thm:main0} by showing that the chain of 1-bit selectors of length $k$ has an XOR game protocol with bias at least $\delta^k$.
This protocol may be regarded as Paw\l owski et al.'s protocol for the pointer jumping function rather than the address function.

\subsection{Proof of Theorem~\ref{thm:main0}}
The main idea comes from Paw\l owski et al.'s protocol~\cite{pawlowski2009information}.
\begin{definition}
The address function $\mathrm{Addr}_n\colon \{0,1\}^{2^n}\times\{0,1\}^n\to\{0,1\}$ is defined by
\begin{equation*}
\mathrm{Addr}_n(x_0,\dotsc,x_{2^n-1},y_1,\dotsc,y_n)=x_y
\end{equation*}
where $y:=\sum_{i=1}^n y_i2^{i-1}$.
\end{definition}
\begin{lemma}\label{lem:addr}
\begin{equation*}
\beta_{\mathrm{NL}(\delta)}(\mathrm{Addr}_1) \ge \delta.
\end{equation*}
\begin{proof}
According to the equation
\begin{equation*}
\mathrm{Addr}_1(x_0, x_1, y_0) = x_0 \oplus y_0(x_0\oplus x_1)
\end{equation*}
Alice and Bob put $x_0\oplus x_1$ and $y_0$ into a nonlocal box and get $a_0$ and $b_0$, respectively.
Then, set $a=x_0\oplus a_0$ and $b=b_0$. This protocol has bias $\delta$.
\end{proof}
\end{lemma}
Paw\l owski et al.\ showed $\beta_{\mathrm{NL}(\delta)}(\mathrm{Addr}_n) \ge \delta^n$ by the iterative application of Lemma~\ref{lem:addr}~\cite{pawlowski2009information, PhysRevA.94.052130}.
For the proof of Theorem~\ref{thm:main0}, we first show the following proposition.
\begin{proposition}\label{prop:main}
Let $X$ and $Y$ be finite sets.
For any $f\colon X\times Y\to \{0,1\}$ and probability distribution $\mu$ on $X\times Y$,
\begin{equation*}
\beta_{\mathrm{NL}(\delta)}(f,\mu) \ge \delta^{C_{1}(f,\mu)}.
\end{equation*}
\end{proposition}
\begin{proof}
We show the proposition by the induction on the communication complexity.
If $C_1(f,\mu)=0$, the proposition trivially holds.
Assume that the proposition holds for $C_1(f,\mu)\le k$.
Let $P$ be the communication protocol for $f$ with the transcript length at most $k+1$.
Assume that at the first step of $P$, Alice sends a bit $A_1(x)$ to Bob.
Let $f_0$ and $f_1$ be the restrictions of $f$ to the rectangles $A_1^{-1}(0)\times Y$ and $A_1^{-1}(1)\times Y$, respectively.
Then, the communication complexity of $f_0$ and $f_1$ are at most $k$.
We can extend the domains of $f_0$ and $f_1$ to $X\times Y$ while the values on the original domains and communication complexities are preserved.
The extended functions are denoted by $\bar{f_0}$ and $\bar{f_1}$.
Then,
\begin{equation*}
f(x,y) = \mathrm{Addr}_1(\bar{f_0}(x,y), \bar{f_1}(x,y), A_1(x)).
\end{equation*}
From the hypothesis of the induction,
there are XOR game protocols for $\bar{f_0}$ and $\bar{f_1}$ with bias at least $\delta^k$.
By applying the XOR game protocols to $\bar{f_0}$ and $\bar{f_1}$, Alice and Bob gets $(a_0, a_1)$ and $(b_0, b_1)$, respectively, such that
$a_0\oplus b_0=f_0(x,y)$ and $a_1\oplus b_1=f_1(x,y)$ with bias $\delta^k$.
Let $e_i:=a_i\oplus b_i \oplus f_i(x,y)$ for $i=0,1$.
Then, we obtain
\begin{align*}
f(x,y) &= \mathrm{Addr}_1(a_0\oplus b_0\oplus e_0, a_1\oplus b_1\oplus e_1, A_1(x))\\
&= 
a_{A_1(x)}
\oplus \mathrm{Addr}_1(b_0, b_1, A_1(x))
\oplus e_{A_1(x)}
\end{align*}
From Lemma~\ref{lem:addr}, Alice and Bob get $a'$ and $b'$ such that $a'\oplus b'=\mathrm{Addr}_1(b_0, b_1, A_1(x))$ with bias $\delta$.
Let $e':=a'\oplus b'\oplus \mathrm{Addr}_1(b_0, b_1, A_1(x))$.
Then,
\begin{align*}
f(x,y) &= 
(a_{A_1(x)} \oplus a') \oplus b'
\oplus (e' \oplus e_{A_1(x)}).
\end{align*}
Let $a:=(a_{A_1(x)} \oplus a')$ and $b:=b'$ be Alice and Bob's final output for the XOR game.
This protocol has bias at least $\delta^{k+1}$.
\end{proof}

We can now straightforwardly show Theorem~\ref{thm:main0}.
\begin{proof}[The proof of Theorem~\ref{thm:main0}]
A deterministic communication protocol corresponding to $C_\rho(f, \mu)$ computes some function $f'$ without error such that $\mathbb{E}_{(x,y)\sim\mu}[f(x,y)\allowbreak f'(x,y)]\ge\rho$.
Hence, the XOR game protocol in Proposition~\ref{prop:main} for $f'$ has bias at least $\rho\delta^{C_\rho(f,\mu)}$ for $f$.
\end{proof}

\begin{remark}
The proof of Proposition~\ref{prop:main} cannot be generalized for quantum communication complexity nor communication complexity with nonlocal boxes straightforwardly since
after Alice uses a nonlocal box, Bob can use the same nonlocal box at most once.
We can generalize the above results to restricted protocols in which Alice and Bob must use common nonlocal boxes (quantum states) in common round.
However, this seems to be restrictive since the standard quantum teleportation is not allowed in the restricted protocols.
\end{remark}

\subsection{Proof of Theorem~\ref{thm:main2}}
Theorem~\ref{thm:main2} is obtained by Proposition~\ref{prop:main} with Buhrman's idea for generalized discrepancy theory~\cite{TCS-040, briet2013multipartite}.
First, we apply the protocol in Proposition~\ref{prop:main} to $f$ by using a randomized communication protocol corresponding to $R_\epsilon(f)$.
This protocol computes $h$ with an XOR of three errors: (a) Error in the computation of the chain of $\mathrm{Addr}_1$ in Proposition~\ref{prop:main}.
(b) Error of the randomized protocol for $f$.
(c) Error from the incoincidence of $f$ and $h$.
The bias of these errors are $\delta^{R_\epsilon(f)}$, $1-2\epsilon$ and $\mathbb{E}_{(x,y)\sim\mu}[f(x,y)h(x,y)]$, respectively.
The error (a) is independent of errors (b) and (c). However, errors (b) and (c) are not independent.
The XOR of errors (b) and (c) is zero with probability at least $\Pr_{(x,y)\sim\mu}(f(x,y) = h(x,y))(1-\epsilon)$, which corresponds to a bias
$(1-\epsilon)\mathbb{E}_{(x,y)\sim\mu}[f(x,y)h(x,y)] - \epsilon$.
Hence, we obtain~\eqref{eq:main2}.

\section{Proof of Theorem~\ref{thm:main1}: Bias of XOR game for XOR functions}
\subsection{XOR game for XOR functions}
Let $g^\oplus(x, y) := g(x\oplus y)$ for any $g\colon \{0,1\}^n\to\{0,1\}$.
A function in this form is called an XOR function.
Let $q\colon\{0,1\}^n\to\mathbb{R}$ be a non-negative function with $\sum_{z} q(z) =2^n$.
The largest bias of XOR game of XOR function $g^\oplus$ on XOR input distribution $2^{-2n}q^\oplus$ can be represented by the largest Fourier amplitude of $g(z)q(z)$.
\begin{lemma}
\begin{equation*}
\beta(g^\oplus, 2^{-2n}q^\oplus) = \max_{S\subseteq[n]} |\widehat{gq}(S)|
=\|\widehat{gq}\|_\infty
\end{equation*}
where $(gq)(z):=g(z)q(z)$.
\end{lemma}
\begin{proof}
When an input distribution is fixed, shared random bits do not help to increase the winning probability of an XOR game.
Hence, without loss of generality, we can assume that Alice and Bob output $a=A(x)$ and $b=B(y)$, respectively where $A$ and $B$ are deterministic boolean functions from $\{0,1\}^n$ to $\{0,1\}$.
Then, we obtain an upper bound of bias of an XOR game for a XOR function (similarly to the proof on the BLR test in~\cite{odonnell2014analysis})
\begin{align*}
\mathbb{E}\left[g(x\oplus y)q(x\oplus y) A(x) B(y)\right]
&= \sum_{S\subseteq [n]} \widehat{gq}(S)\widehat{A}(S)\widehat{B}(S)\\
&\le \max_{S\subseteq [n]} |\widehat{gq}(S)| \sum_{S\subseteq [n]} |\widehat{A}(S)\widehat{B}(S)|\\
&\le \max_{S\subseteq [n]} |\widehat{gq}(S)|.
\end{align*}
The last inequality is obtained by the Cauchy--Scwartz inequality and $\sum_{S\subseteq[n]} \widehat{A}(S)^2 = \sum_{S\subseteq [n]} \widehat{B}(S)^2 = 1$.
Let $S^*:=\mathrm{argmax}_{S\subseteq[n]} |\widehat{gq}(S)|$.
This upper bound can be achieved by $A(x)=\bigoplus_{i\in S^*} x_i$, $B(y)=\mathrm{sign}(\widehat{gq}(S^*))\oplus \bigoplus_{i\in S^*} y_i$ where
$\mathrm{sign}(x)$ is 0 if $x\ge0$ and 1 otherwise.
\end{proof}

For general input distribution, Alice and Bob easily make the distribution XOR by using shared random bits
since $g(x\oplus y)= g((x\oplus r)\oplus (y\oplus r))$.
The probability distribution of $(x', y'):=(x\oplus r, y\oplus r)$ is an XOR function.
In other word, the worst case input distribution must be an XOR function.

\begin{lemma}
\begin{equation*}
\beta(g^\oplus) = \min_{q} \|\widehat{gq}\|_\infty.
\end{equation*}
\end{lemma}
\begin{remark}
If we consider worst boolean function $g$,
we obtain the lower bound
\begin{equation}
\min_{g,q}\max_{S\subseteq[n]} |\widehat{gq}(S)|
\ge
\sqrt{\frac1{2^n}\mathbb{E}[q(z)^2]}
\ge 2^{-\frac{n}2}
\label{eq:worst}
\end{equation}
from Parseval's identity $\sum_{S} \widehat{gq}(S)^2 = \mathbb{E}[q(z)^2]$.
This lower bound can be achieved by bent functions and the uniform input distribution~\cite{odonnell2014analysis}.
This lower bound $2^{-n/2}$ is not applicable for general non-XOR function.
If $n=1$ and $f(x,y)=x\wedge y$, i.e., CHSH game, the maximum bias $\beta(\mathrm{AND})=1/2<1/\sqrt{2}$.
For general non-XOR functions, Littlewood's 4/3 inequality gives a lower bound $2^{-(n+1)/2}$~\cite{doi:10.1063/1.4769269}.
\end{remark}

Linden et al.\ showed $\beta^*(g^\oplus, 2^{-2n}q^\oplus) = \beta(g^\oplus, 2^{-2n}q^\oplus)$ for any $g$ and $q$~\cite{PhysRevLett.99.180502}.
We give another proof using Tsirelson's characterization and Fourier analysis in Appendix~\ref{apx:qxor}.

\subsection{Equality function}
The negation of the equality function $\mathrm{EQ}_n$ is an XOR function $\mathrm{OR}_n^\oplus$.
\begin{lemma}
\begin{equation}
\lim_{n\to\infty}\min_{q}\left\|\widehat{\mathrm{OR}_nq}\right\|_\infty
=
\frac13.
\end{equation}
\end{lemma}
\begin{proof}
First, we show $\lim_{n\to\infty}\min_{q}\|\widehat{\mathrm{OR}_nq}\|_\infty\le 1/3$.
We consider an optimization of the input distribution $\mu(z):=2^{-n}q(z)$ only among a class
\begin{align*}
\mu(z)=\begin{cases}
\lambda,& w(z) = 0\\
\frac{1-\lambda}{\binom{n}{2}},& w(z) = 2\\
0,& \text{otherwise}
\end{cases}
\end{align*}
where $w(z)$ denotes the Hamming weight of $z\in\{0,1\}^n$ and $\lambda\in[0,1]$ is a parameter.
Then,
\begin{align*}
&\min_{q}\max_{S\subseteq[n]} \left|\mathbb{E}\left[\mathrm{OR}_n(z)q(z)\prod_{i\in S} z_i\right]\right|\\
&\le\min_{\lambda}\max_{S\subseteq[n]} \left| \lambda - (1-\lambda)\left(\frac{(n-|S|)(n-|S|-1)}{n(n-1)} - 2\frac{(n-|S|)|S|}{n(n-1)} + \frac{|S|(|S|-1)}{n(n-1)}\right)\right|\\
&=\min_{\lambda}\max_{k=0,1,\dotsc,n} \left| \lambda - (1-\lambda)\left(\frac{n^2-4nk+4k^2-n}{n(n-1)}\right)\right|\\
&=\min_{\lambda}\max_{k=0,1,\dotsc,n} \left| 2\lambda - 1 + (1-\lambda)\left(\frac{4k(n-k)}{n(n-1)}\right)\right|\\
&\le\min_{\lambda}\max \left\{1-2\lambda, 2\lambda - 1 + (1-\lambda)\frac{n}{n-1}\right\}\\
&= \min_{\lambda}\max \left\{1-2\lambda, \lambda+\Theta(n^{-1})\right\}
\end{align*}
By solving $1-2\lambda = \lambda$, we obtain $\lambda=1/3$.

Next, we show $\min_{q}\|\widehat{\mathrm{OR}_nq}\|_\infty\ge 1/3$ for any $n$.
From the minimax principle, it is sufficient to show a randomized protocol for the XOR game for the negation of the equality function with bias at least $1/3$ for any input $(x,y)\in\{0,1\}^n\times\{0,1\}^n$.
We first consider two protocols for XOR game of the negation of the equality function, and then take a probabilistic mixture of them.
The first protocol always answer 1, i.e., $A(x)=0$, $B(y)=1$.
The second protocol uses shared random bits $r\in\{0,1\}^n$ and take inner products with inputs, i.e., $A(x)=\langle x, r\rangle$, $B(y)=\langle y, r\rangle$.
The bias of the first protocol is $1-2\lambda$ where $\lambda$ denotes the probability of $x=y$.
The bias of the second protocol is $\lambda$~\cite{kushilevitz1997communication}.
Hence, if we choose the first protocol with probability 1/3 and choose the second protocol with probability 2/3,
we get the bias $1/3$.
\end{proof}
Next, we show the worst i.i.d. input distribution for the equality function.
Similarly to the general case, we can assume that an input distribution for each bit is an XOR function.
\begin{lemma}
For $n\ge 2$,
\begin{equation*}
\min_{\nu}\left\|\widehat{\mathrm{OR}_n(2\nu)^{\otimes n}}\right\|_\infty
=
1-2 \lambda^{*n}
\end{equation*}
where $\nu$ denotes a distribution on $\{0,1\}$, and where $\lambda^*$ denotes the unique root in $[1/2,1]$ of
\begin{equation}\label{eq:polyeq}
4\lambda^n-(2\lambda-1)^n-1=0.
\end{equation}
Furthermore,
\begin{equation*}
\lim_{n\to\infty}
\min_{\nu}\left\|\widehat{\mathrm{OR}_n(2\nu)^{\otimes n}}\right\|_\infty
=
2\sqrt{3}-3.
\end{equation*}
\end{lemma}
\begin{proof}
\begin{align*}
\min_{\nu}\max_{S\subseteq[n]} |\widehat{\mathrm{OR}_n(2\nu)^{\otimes n}}(S)|
&=
\min_{\nu}\max_{S\subseteq[n]} \left|\mathbb{E}\left[\mathrm{OR}(z)(2\nu)^{\otimes n}(z)\prod_{i\in S} z_i\right]\right|\\
&=
\min_{\nu}\max_{S\subseteq[n]} \left|\mathbb{E}\left[\left(2\mathbb{I}\left\{z=+1\right\}-1\right)(2\nu)^{\otimes n}(z)\prod_{i\in S} z_i\right]\right|\\
&=
\min_{\nu}\max_{S\subseteq[n]} \left| 2 \nu(0)^n - (\nu(0)-\nu(1))^{|S|}\right|\\
&=
\min_{\nu}\max\left\{1- 2 \nu(0)^n, 2 \nu(0)^n - (\nu(0)-\nu(1))^{n}\right\}\\
&=
\min_{\nu}\max\left\{1- 2 \nu(0)^n, 2 \nu(0)^n - (2\nu(0)-1)^{n}\right\}
\end{align*}
Let $u_n(\lambda):=1-2\lambda^n$ and $v_n(\lambda):=2\lambda^n-(2\lambda-1)^n$.
It is easy to see that $u_n(\lambda)$ is monotonically decreasing for $\lambda\in[0,1]$, and that $v_n(\lambda)$ is monotonically increasing for $\lambda\in[1/2,1]$.
For $n\ge 2$, $u_n(1/2)\ge v_n(1/2)$.
Hence, for $n\ge 2$, we obtain
\begin{equation*}
\min_{\nu}\max\left\{1- 2 \nu(0)^n, 2 \nu(0)^n - (2\nu(0)-1)^{n}\right\}
=1-2\lambda^{*n}
\end{equation*}
where $\lambda^*$ is the unique root in $[1/2,1]$ of~\eqref{eq:polyeq}.
Let $\lambda_n=1-c/n$ for some constant $c>0$.
Then, we obtain
\begin{align*}
\lim_{n\to\infty}
4\lambda_n^n-(2\lambda_n-1)^n-1=
4\mathrm{e}^{-c} - \mathrm{e}^{-2c} - 1
\end{align*}
Let $\varphi = \mathrm{e}^{-c} < 1$.
\begin{align*}
4\varphi - \varphi^2 - 1 = 0
\iff
\varphi^2 - 4\varphi + 1 = 0
\iff
\varphi = 2 - \sqrt{3}
\end{align*}
Then, we obtain
\begin{equation*}
\lim_{n\to\infty} 1-2\lambda^{*n}
= 1-2\mathrm{e}^{-c}
= 1- 2(2-\sqrt{3}) = 2\sqrt{3}-3.
\qedhere
\end{equation*}
\end{proof}

\section{Remarks}
\subsection{Linial and Shraibman's lower bound for XOR functions}
When we apply Linial and Shraibman's lower bound~\eqref{eq:ls} for XOR function $g^\oplus$, we obtain
\begin{align*}
R_\epsilon(g^\oplus)\ge 2  \log\max_{h, \mu}\frac{(1-\epsilon)\mathbb{E}_{(x,y)\sim \mu}[g^\oplus(x,y)h(x,y)]-\epsilon}{\beta^*(h,\mu)}.
\end{align*}
For $\epsilon=0$, we obtain
\begin{align}
R_0(g^\oplus)&\ge 2  \log\max_{h, q}\frac{\mathbb{E}_{(x,y)\sim 2^{-2n}q^\oplus}[g^\oplus(x,y)h^\oplus(x,y)]}{\beta^*(h^\oplus,2^{-2n}q^\oplus)}\nonumber\\
&= 2  \log\max_{h,q} \frac{\sum_S \widehat{g}(S)\widehat{hq}(S)}{\|\widehat{hq}\|_\infty}
=2\log \|\widehat{g}\|_1.
\label{eq:inf1}
\end{align}
This fact was shown in~\cite{linial2009lower, TCS-040}.
For $\epsilon\in(0,1/2)$, Linial and Shraibman's lower bound can be written as
\begin{align}\label{eq:apx}
R_\epsilon(g^\oplus)&\ge 2\log\left[(1-\epsilon)\|\widehat{g}\|_1^{\frac{\epsilon}{1-\epsilon}}\right]
\end{align}
where
\begin{align*}
\|\widehat{g}\|_1^\epsilon&:=\min\left\{\|\widehat{h}\|_1\mid h\colon \{+1,-1\}^n\to\mathbb{R}, |h(x)-g(x)|\le \epsilon\, \forall x\right\}
\end{align*}
is called approximate Fourier $\ell_1$ norm~\cite{linial2009lower}.
Obviously, $\|\widehat{g}\|_1^\epsilon\le (1-\epsilon)\|\widehat{g}\|_1$.
However, no lower bound of the approximate Fourier $\ell_1$ norm by the exact Fourier $\ell_1$ norm has been known for general $\epsilon\in(0,1)$.
Here, we give a necessary and sufficient condition for the equality $\|\widehat{g}\|_1^\epsilon = (1-\epsilon)\|\widehat{g}\|_1$.
\begin{definition}
For $f\colon\{+1,-1\}^n\to\{+1,-1\}$,
\begin{align*}
f^* &:= \bigl\{g\colon\{+1,-1\}^n\to[-1,+1] \mid g(x) = \sign(\widehat{f}(S_x)), \forall x, \widehat{f}(S_x)\ne 0\bigr\}\\
f^{**} &:= \bigl\{h\colon\{+1,-1\}^n\to[-1,+1] \mid \exists g\in f^*, h(x) = \sign(\widehat{g}(S_x)), \forall x, \widehat{g}(S_x)\ne 0\bigr\}.
\end{align*}
where $S_x:=\{i\mid x_i=-1\}$ and
$\sign(r):=r/|r|$ for $r\ne 0$.
\end{definition}

The complementary slackness condition gives the following theorem.
\begin{theorem}\label{thm:kkt}
$\|\widehat{f}\|_1^\epsilon = (1-\epsilon)\|\widehat{f}\|_1$
if and only if $f\in f^{**}$.
\end{theorem}
The proof is in Appendix~\ref{apx:kkt}.
\begin{remark}
Exhaustive search on computer shows that
all of 256 boolean functions on 3 variables,
51200 of 65536 boolean functions on 4 variables and
at least 2839187456 of 4294967296 boolean functions on 5 variables satisfy
$f\in f^{**}$.
\end{remark}
For $\mathrm{OR}_n$, $\mathrm{OR}_n^*$ only includes the negation of $\mathrm{OR}_n$.
Hence, $\mathrm{OR}_n^{**}$ only includes $\mathrm{OR}_n$.
Then, Theorem~\ref{thm:kkt} gives $\|\widehat{\mathrm{OR}_n}\|_1^\epsilon=(1-\epsilon)\|\widehat{\mathrm{OR}_n}\|_1$ so that
\begin{equation*}
R_\epsilon(\mathrm{EQ}_n)\ge 2\log\left[(1-2\epsilon)\|\widehat{\mathrm{OR}}\|_1\right]=2\log\left[(1-2\epsilon)(3-2^{-n+2})\right]
\end{equation*}
from~\eqref{eq:apx}.
In general, it is difficult to lower bound the approximate Fourier $\ell_1$ norm.
The following theorem gives a simple lower bound which may be useful for some case.
\begin{theorem}
\begin{equation*}
R_\epsilon(g^\oplus) \ge 2\log\left((1-\epsilon)\|\widehat{g}\|_1 - \epsilon \|\widehat{g^*}\|_1\right).
\end{equation*}
where
\begin{equation*}
\|\widehat{g^*}\|_1 := \min\{\|\widehat{h}\|_1\mid h\in g^*\}.
\end{equation*}
\end{theorem}
\begin{proof}
From Linial and Shraibman's bound for XOR functions,
\begin{align*}
R_\epsilon(g^\oplus)&\ge 2  \log\max_{h,q} \frac{(1-\epsilon) \sum_S \widehat{g}(S)\widehat{hq}(S) - \epsilon}{\|\widehat{hq}\|_\infty}.
\end{align*}
Here $hq\colon\{+1,-1\}^n\to\mathbb{R}$ is an arbitrary function satisfying $\sum_{z\in\{+1,-1\}^n}|hq(z)| = 2^n$.
For some $G\in g^*$,
we choose $hq$ such that
\begin{equation*}
\widehat{hq}(S) =  \|\widehat{hq}\|_\infty\, G(S_x).
\end{equation*}
Then,
\begin{align*}
hq(z) &= \sum_{S\subseteq [n]}  \|\widehat{hq}\|_\infty\, G(S_x) \prod_{i\in S} z_i\\
&= \|\widehat{hq}\|_\infty\,  2^n \widehat{G}(S_z).
\end{align*}
Hence,
\begin{align*}
\|\widehat{hq}\|_\infty = \frac1{\|\widehat{G}\|_1}.
\end{align*}
The theorem is obtained.
\end{proof}
From the above proof, we obtain $\|\widehat{g^*}\|_1\ge\|\widehat{g}\|_1$.
From the Cauchy--Schwartz inequality, we obtain $\|\widehat{g^*}\|_1\le \sqrt{\|\widehat{g}\|_0}$, which gives a weaker bound
$R_\epsilon(g^\oplus) \ge 2\log\left((1-\epsilon)\|\widehat{g}\|_1 - \epsilon \sqrt{\|\widehat{g}\|_0}\right)$ shown in~\cite{linial2009lower}.

\subsection{Relationship with information complexity of the equality function}
Braverman and Rao showed that information complexity is equal to (direct product-)amortized communication complexity~\cite{braverman2014information}.
The information complexity is defined by
\begin{equation*}
\mathrm{IC}(f, \mu) := \min_{\pi} I(X; \pi(X, Y)\mid Y) + I(Y; \pi(X, Y)\mid X)
\end{equation*}
where $I$ denotes the mutual information, where $\pi$ denotes a protocol, and where $\pi(X,Y)$ denotes
the public randomness and a transcript when $X$ and $Y$ are given to a protocol $\pi$ (See~\cite{braverman2014information} and~\cite{braverman2015interactive} for details).
In~\cite{braverman2015interactive}, the information complexity of the equality function is upper bounded by 9.
In this section, this upper bound is improved to $2\log 5\approx 4.64$.
We consider a particular protocol introduced in~\cite{braverman2015interactive}.
Alice and Bob use a shared invertible random matrix $A$, whose $i$-th row is denoted by $a_i$.
At $i$-th step Alice and Bob send $\langle x, a_i\rangle$ and $\langle y,a_i\rangle$ to each other, respectively.
If $\langle x, a_i\rangle \ne \langle y, a_i\rangle$, the protocol terminates and output 0.
If $\langle x, a_i\rangle = \langle y, a_i\rangle$ for all $i=1,\dotsc,n$, the protocol output 1.
Let $Z$ be a random variable taking a value 1 if $x=y$ and 0 if $x\ne y$.
Assume $\Pr(Z=1)=\lambda$.
Then, we obtain
\begin{align*}
I(X;\pi(X,Y)\mid Y)&=
I(X,Z;\pi(X,Y)\mid Y)\\
&= I(Z;\pi(X,Y)\mid Y) + I(X;\pi(X,Y)\mid Y, Z)\\
&= H(Z\mid Y) + I(X;\pi(X,Y)\mid Y, Z)\\
&\le h(\lambda) + (1-\lambda) \sum_{i\ge 1}\frac1{2^i}i\\
&= h(\lambda) + 2(1-\lambda)
\end{align*}
where $h(\lambda):=-\lambda\log\lambda-(1-\lambda)\log(1-\lambda)$.
Here, $h(\lambda)+2(1-\lambda)$ is maximized at $\lambda=1/5$ with the maximum $\log 5$.
Hence, $2\log 5$ is an upper bound of the information complexity.

Here, any protocol obviously requires at least $h(\lambda)$ bits for each direction in average.
If we assume that at least extra 1 bit is required  for each direction if $x\ne y$, we obtain a lower bound $h(\lambda)+1-\lambda$ for each direction,
which is maximized at $\lambda=1/3$ with the maximum $\log 3$.
Hence, the lower bound $2\log 3$ of XOR-amortized communication complexity in Theorem~\ref{thm:exm} intuitively means that
``each of the equality problems must be solved by using at least $h(\lambda) + 1-\lambda$ bits for each direction''.

\section*{Acknowledgment}
The author would like to thank Mark Braverman for the insight that the lower bound 2 for XOR-amortized communication complexity of the equality function seems to be trivial.
The author would like to thank Ryan O'Donnell for insightful discussions.

\bibliographystyle{plain}
\bibliography{biblio}

\appendix
\section{Quantum XOR game for XOR functions}\label{apx:qxor}

\begin{lemma}[\cite{PhysRevLett.99.180502}]
$\beta^*(g^\oplus, 2^{-2n}q^\oplus) = \beta(g^\oplus, 2^{-2n}q^\oplus)$ for any $g$ and $q$.
\end{lemma}
\begin{proof}
From Tsirelson's characterization~\cite{tsirel1987quantum}, $\beta^*(g^\oplus, 2^{-n}q^\oplus)$ is equal to
\begin{equation*}
\max_{V, W}\mathbb{E}\left[g(x\oplus y)q(x\oplus y) \langle V(x), W(y)\rangle\right]
\end{equation*}
where $V(x)$ and $W(y)$ are unit vectors on $\mathbb{R}$ of dimension $2n$ for all $x,\,y\in\{0,1\}^n$.
Then, we obtain
\begin{align*}
&\mathbb{E}\left[g(x\oplus y)q(x\oplus y) \langle V(x), W(y)\rangle\right]
=
\mathbb{E}\left[g(x\oplus y)q(x\oplus y) \sum_{i=1}^{2n} V_i(x) W_i(y)\right]\\
&= \sum_{S\subseteq [n]}\widehat{gq}(S) \sum_{i=1}^{2n}\widehat{V}_i(S)\widehat{W}_i(S)\\
&\le \max_{S\subseteq [n]}|\widehat{gq}(S)| \sum_{S\subseteq [n]}\sum_{i=1}^{2n}\left|\widehat{V}_i(S)\widehat{W}_i(S)\right|\\
&\le \max_{S\subseteq [n]}|\widehat{gq}(S)| \sqrt{\left(\sum_{S\subseteq [n]}\sum_{i=1}^{2n}\widehat{V}_i(S)^2\right)\left(\sum_{S\subseteq [n]}\sum_{i=1}^{2n}\widehat{W}_i(S)^2\right)}\\
&= \max_{S\subseteq [n]}|\widehat{gq}(S)| \sqrt{\left(\frac1{2^n}\sum_{i=1}^{2n} \sum_{x\in\{0,1\}^n} V_i(x)^2\right)\left(\frac1{2^n}\sum_{i=1}^{2n} \sum_{y\in\{0,1\}^n} W_i(y)^2\right)}\\
&= \max_{S\subseteq [n]}|\widehat{gq}(S)| \sqrt{\left(\sum_{x\in\{0,1\}^n} \frac1{2^n}\right)\left(\sum_{y\in\{0,1\}^n} \frac1{2^n}\right)}\\
&= \max_{S\subseteq [n]}|\widehat{gq}(S)|.
\qedhere
\end{align*}
\end{proof}

\section{Proof of Theorem~\ref{thm:kkt}}\label{apx:kkt}
Let $\mathcal{\chi}_S(x):=\prod_{i\in S}x_i$.
Then, $\|\widehat{f}\|_1^\epsilon$ is the solution of the following optimization problem.
\begin{align*}
\min:& \sum_{S\subseteq[n]} \left|\frac1{2^n}\sum_x g(x) \chi_S(x)\right|\\
\text{subject to}:& |g(x) - f(x)| \le \epsilon,\qquad \forall x\in\{+1,-1\}^n.
\end{align*}
\begin{lemma}\label{lem:lp}
$g\colon\{+1,-1\}^n\to\mathbb{R}$ is optimal of the above optimization problem if and only if
there exists $h\in g^*$ such that
\begin{align*}
g(x) =
\begin{cases}
 f(x) - \epsilon,& \text{if } \widehat{h}(S_x)>0 \\
 f(x) + \epsilon,& \text{if } \widehat{h}(S_x)<0 \\
 \in [f(x)-\epsilon, f(x)+\epsilon],& \text{otherwise.}
\end{cases}
\end{align*}
\end{lemma}
\begin{proof}
The optimization problem is equivalent to the following linear program.
\begin{align*}
\min:& \sum_{S\subseteq[n]} \left(\widehat{g}_+(S) + \widehat{g}_-(S)\right)\\
\text{subject to}: 
&\, \widehat{g}_+(S) - \widehat{g}_-(S) = \frac1{2^n}\sum_x g(x) \chi_S(x),\qquad\forall S\subseteq[n]\\
&\, g(x) - f(x) \le \epsilon,\qquad \forall x\in\{+1,-1\}^n\\
&\, f(x) - g(x) \le \epsilon,\qquad \forall x\in\{+1,-1\}^n\\
&\, \widehat{g}_+(S) \ge 0,\qquad \forall S\subseteq[n]\\
&\, \widehat{g}_-(S) \ge 0,\qquad \forall S\subseteq[n]
\end{align*}
where $\widehat{g}_+(S)$ and $\widehat{g}_-(S)$ are variables of the above linear program.
The Lagrangian of this optimization problem is
\begin{align*}
&\mathcal{L}(g, \widehat{g}_+, \widehat{g}_-, c, \lambda, \rho,\mu_+,\mu_-):=
\sum_{S\subseteq[n]} \left(\widehat{g}_+(S) + \widehat{g}_-(S)\right)\\
&\quad - \sum_S c_S\left(\widehat{g}_+(S) - \widehat{g}_-(S) - \frac1{2^n}\sum_x g(x) \chi_S(x)\right)\\
&\quad +\sum_x \lambda(x)\left(g(x)-f(x)-\epsilon\right)
+\sum_x \rho(x)\left(f(x)-g(x)-\epsilon\right)\\
&\quad - \sum_S \mu_+(S)\widehat{g}_+(S) - \sum_S \mu_-(S)\widehat{g}_-(S)
\end{align*}
where the dual feasibility condition is
\begin{align*}
\lambda(x)&\ge 0,& \rho(x)&\ge 0,\qquad \forall x\in\{+1,-1\}^n\\
\mu_+(S)&\ge 0,& \mu_-(S)&\ge 0,\qquad \forall S\subseteq[n].
\end{align*}
First, we consider necessary conditions for optimal primal and dual solutions.
The partial derivatives of Lagrangian with respect to primal variables are
\begin{align*}
\frac{\partial \mathcal{L}}{\partial g(x)}&=
\lambda(x)-\rho(x)+\frac1{2^n}\sum_S c_S \chi_S(x)\\
\frac{\partial \mathcal{L}}{\partial \widehat{g}_+(S)}&=
1-c_S-\mu_+(S)\\
\frac{\partial \mathcal{L}}{\partial \widehat{g}_-(S)}&=
1+c_S-\mu_-(S).
\end{align*}
All of them must be zero at any optimal primal dual pair.
Hence, $c_S\in[-1,+1]$ for any $S\subseteq[n]$.
A complementary slackness condition says
\begin{align*}
\widehat{g}_+(S)>0 \,\Rightarrow\, \mu_+(S) = 0 \iff c_S = +1\\
\widehat{g}_-(S)>0 \,\Rightarrow\, \mu_-(S) = 0 \iff c_S = -1.
\end{align*}
Since for any primal optimal solution,
 $\widehat{g}_+(S)>0\iff\widehat{g}(S)>0$ and
 $\widehat{g}_-(S)>0\iff\widehat{g}(S)<0$,
the above condition means that $h(x):=c_{S_x}$ must be a member of $g^*$.
Furthermore,
from $\partial\mathcal{L}/\partial g(x)=0$,
\begin{align*}
\widehat{h}(S_x)>0 \,\Rightarrow\, \rho(x) >0 \Rightarrow g(x) = f(x)-\epsilon\\
\widehat{h}(S_x)<0 \,\Rightarrow\, \lambda(x) >0 \Rightarrow g(x) = f(x)+\epsilon
\end{align*}
This shows the one direction of this Theorem.

Conversely,
if there exists $h\in g^*$ satisfying the condition, we can choose values
\begin{align*}
\widehat{g}_+(S) &= \max\{0,\widehat{g}(S)\},&
\widehat{g}_-(S) &= \max\{0,-\widehat{g}(S)\}\\
c_{S_x} &= h(x)\\
\lambda(x) &= \max\{0, -\widehat{h}(S_x)\},&
\rho(x) &= \max\{0, \widehat{h}(S_x)\}\\
\mu_+(S) &= 1- c_S,&
\mu_-(S) &= 1+ c_S
\end{align*}
which satisfy all complementary slackness conditions, and hence are optimal.
\end{proof}

Theorem~\ref{thm:kkt} is immediately obtained from Lemma~\ref{lem:lp}.

\end{document}